\documentclass[%
reprint,
amsmath,amssymb,
aps,
]{revtex4-2}

\usepackage{graphicx}
\usepackage{dcolumn}
\usepackage{bm}
\usepackage[colorlinks=true,linkcolor=black]{hyperref}

\usepackage{derivative}
\usepackage{enumitem}
\usepackage{array}
\usepackage{amsthm}
\renewcommand\vec{\mathbf} 
\newcolumntype{L}[1]{>{\raggedright\arraybackslash}p{#1}}

\newtheorem{lem}{Lemma}
\newtheorem{cor}{Corollary}
\newtheorem{theorem}{Theorem}

\newtheoremstyle{axiom}
{4pt}
{2pt}
{}
{1.0em}
{}
{.}
{0.3em}
{\thmname{#1}~\thmnumber{(#2)}\thmnote{ #3}}
{}

\theoremstyle{axiom}
\newtheorem{axm}{Axiom}
\newtheorem{defn}{Definition}
\usepackage{xr}
\makeatletter
\newcommand*{\addFileDependency}[1]{
\typeout{(#1)}
\@addtofilelist{#1}
\IfFileExists{#1}{}{\typeout{No file #1.}}
}\makeatother
\newcommand*{\myexternaldocument}[1]{%
\externaldocument{#1}%
\addFileDependency{#1.tex}%
\addFileDependency{#1.aux}%
}
\myexternaldocument{si}

\begin{document}

\title{A Local Gauge-Covariant Formulation of Classical Dynamics}
\author{Gunjan Auti}
\email{gunjanauti@thml.t.u-tokyo.ac.jp}
\affiliation{Department of Mechanical Engineering, The University of Tokyo, 7-3-1 Hongo, Bunkyo-ku, Tokyo 113-8656, Japan}

\author{Gouhei Tanaka}
\affiliation{Program of Computational Intelligence, Graduate School of Engineering, Nagoya Institute of Technology, Nagoya 466-8555, Japan}
\affiliation{International Research Center for Neurointelligence, Institutes for Advanced Study, The University of Tokyo, Tokyo 113-0033, Japan}

\author{Hirofumi Daiguji}
\affiliation{Department of Mechanical Engineering, The University of Tokyo, 7-3-1 Hongo, Bunkyo-ku, Tokyo 113-8656, Japan}

\date{\today}

\begin{abstract}
Classical dynamical laws are conventionally formulated as closed evolution equations defined on fixed geometric backgrounds and a global time parameter. We develop a formulation in which neither prescribed evolution laws nor an external clock are assumed \emph{a priori}. Grounded in the principles of conservation, locality of interaction, and independent local frame freedom, the framework treats spatial geometry as a relational structure that may evolve together with the state.
We introduce a notion of local incompatibility defined as the covariant difference between neighboring states under a dynamical transport connection. Because the transport relations are not fixed, restoring compatibility requires the joint adaptation of both state variables and transport geometry. We show that locality, gauge covariance, and coercivity strongly restrict the admissible form of this incompatibility and lead to a simple, globally additive, gauge-invariant quadratic measure of mismatch. Admissible dynamics are then defined as the asynchronous, finite-rate relaxation of this measure, without assuming a predefined action principle. A global time description appears only as an effective coarse-grained limit of this local relaxation process.
In appropriate limits, the resulting compatibility-restoration dynamics recovers familiar continuum equations, including diffusion, incompressible Navier–Stokes, and the Amp\`ere–Maxwell relation. In this sense, dynamics arises from the coupled evolution of state and transport geometry toward local gauge consistency. The formulation provides a constructive framework in which effective physical laws emerge from local relational constraints.
\end{abstract}

\maketitle

\section{Introduction}
At the most basic level, we usually define a physical system as some kind of ``stuff" spread out in space, for example, mass in a fluid, charge in a material, or momentum in motion. We describe this using a state $U$, which simply tells us how much of that stuff is present at each location. But this familiar picture quietly assumes something strong: that the space and the state are separate. The space is taken as a fixed background, and it already tells us how to compare what is happening at one location with what is happening at another. These rules of comparison are assumed to hold everywhere in the system, and to apply simultaneously at each moment in time [Fig.~\ref{fig:figure1}(a)].

In this work, we step away from this assumption. Instead of taking the comparison structure as given, we allow it to evolve along with the state itself [Fig.~\ref{fig:figure1}(b)]. In this view, both the relations between neighboring regions and the states defined on them evolve through local interaction.

\begin{figure*}
    \centering
    \includegraphics[width=\textwidth]{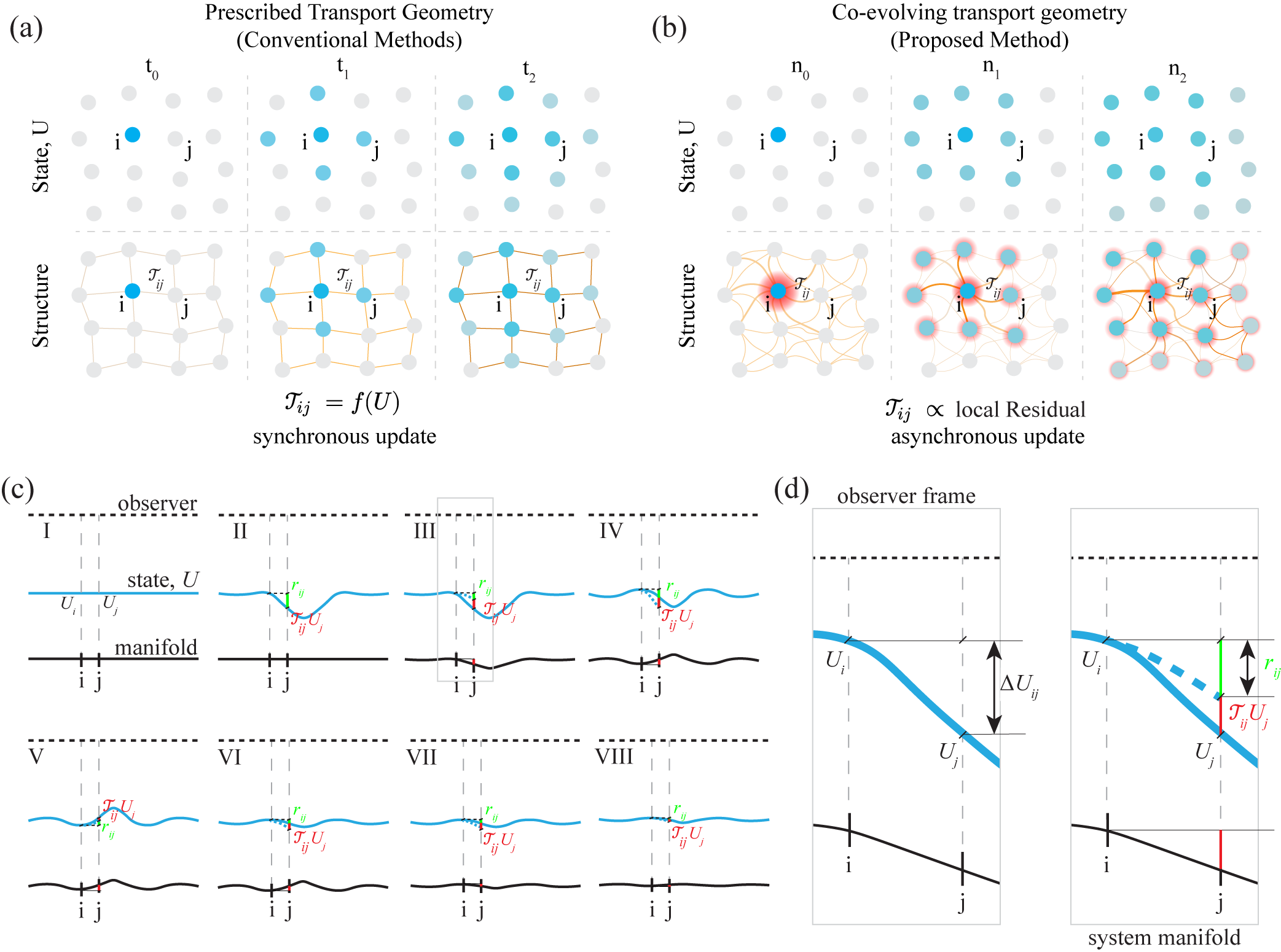}
    \caption{\textbf{Conceptual overview of the formulation based on local inconsistency and dynamical transport geometry.} In both (a) and (b), the system is represented as a graph of local regions (nodes) connected by transport relations $\mathcal T_{ij}$ (edges). The color of each node represents the local state $U_i$, and the panels labeled 0, 1, and 2 indicate successive stages of evolution from the same initial condition. The lower row in each panel shows the corresponding transport structure used to compare neighboring states.
    (a) In standard formulations, the relations used to compare neighboring regions, $\mathcal T_{ij}$, are prescribed a priori. The transport geometry evolves globally and synchronously $(t_0, t_1, t_2)$ according to the same prescribed function of the state, $f(U)$, so that only the local state changes while the comparison structure remains fixed.
    (b) In the present formulation, the transport relations are treated as dynamical variables. Both the state and the transport geometry adapt asynchronously $(n_0, n_1, n_2)$ in response to locally measured inconsistency (residual), shown by the red hue around the nodes in the structural graph, allowing the effective comparison structure between neighboring regions to evolve during the dynamics. Different colors of the connections indicate different transport rules, different thickness denotes the strength (weight) of the connection, and different paths represent changes in the transport geometry (curvature or trajectory) between neighboring regions.
    (c) Schematic illustration of the coupled evolution of the state $U$ (blue) and the transport geometry encoded by the system manifold (black), observer (dashed black) can only observe the state $U$. At each panel, the local state $U_i$ is compared with the neighboring state $U_j$  expressed in the local frame via the current transport $\mathcal T_{ij}$. Incompatibility, $r_{ij}$, is present whenever the transported state $\mathcal T_{ij}U_j$ (dashed blue) fails to coincide with $U_i$, independent of whether the mismatch is initiated by changes in the state or in the geometry. Such incompatibilities drive subsequent adjustments of both $U$ and $\mathcal T$, leading to mutual relaxation without reference to an external time parameter.
    (d) Comparison with an observer-frame description. In a fixed flat chart, only the raw difference $\Delta U_{ij} = U_i - U_j$ is accessible. In the system description, meaningful comparison requires transport along the manifold, and the intrinsic incompatibility is given by the covariant mismatch $r_{ij} \propto (U_i - \mathcal T_{ij}U_j)$.}
    \label{fig:figure1}
\end{figure*}

At the resolved level, physical interactions are local: each region exchanges conserved quantities only with its neighbors and responds only to information available within that neighborhood \cite{einstein1935can, maxwell1865viii, goldstein1950classical}. Because physical influence propagates at finite speed, the relations needed to compare neighboring states are not automatically guaranteed to be globally coordinated or instantaneously available. Standard dynamical formulations, however, usually begin by assuming exactly such a coordination. One postulates an action functional
\begin{equation}
S[U] = \int L(U,\dot U, t)\, dt,
\end{equation}
and physical trajectories are obtained by requiring stationarity of the action. The resulting Euler–Lagrange equations produce closed evolution laws of the form
\begin{equation}
\frac{\partial U}{\partial t} = F(U,\Theta),
\label{eq:eq_one}
\end{equation}
where $F$ encodes forces, transport, and constraints, and $\Theta$ denotes auxiliary assumptions or parameters \cite{poincare1899methodes, goldstein1950classical, strogatz2024nonlinear}. In such formulations, the existence of a globally defined time, a known Lagrangian, and a fixed geometric structure are assumed from the outset. Consequently, the relations that define how neighboring states are compared and how transport occurs are specified independently of the evolving state itself, effectively separating the state from the structure that governs its interaction.

Accordingly, in these formulations neighboring states can be meaningfully compared through a predefined relations. In geometric terms, this means that the rule used to relate neighboring states is assumed in advance, typically through a prescribed connection or symmetry structure that specifies how local frames are compared \cite{yang1954conservation}. Similarly, in non-equilibrium thermodynamics and gradient-flow formulations, the structure used to compare neighboring states, along with the associated transport coefficients, is specified \emph{a priori} \cite{onsager1931reciprocal,onsager1931reciprocal2, prigogine1967symmetry}. More broadly, geometric and data-driven approaches adopt fixed representation spaces, for example, the manifolds, lattices, or feature spaces, within which neighboring states are compared through prescribed relations.

These constructions provide consistent descriptions when the admissible class of transport relations is known. However, in many far-from-equilibrium and multiscale systems, the effective relations governing transport and coordination depend on unresolved processes, local history, or evolving interactions. Examples include active suspensions and pattern-forming systems with multiple intrinsic timescales \cite{cross1993pattern, ramaswamy2010mechanics, marchetti2013hydrodynamics, saintillan2013active}; under-resolved or weakly compressible flows, where missing information leads to non-Markovian or optimal-prediction descriptions \cite{chorin2000optimal}; and electrochemical or phase-transforming systems, where local kinetics determines coupling between neighboring regions \cite{bazant2009towards}. In such regimes, the relations required to compare and transport local states are not available \emph{a priori}, but must instead be established through interaction. Consistency between neighboring regions is therefore not a structural property of the model, but a dynamical requirement arising from the mutual coordination of local descriptions.

Related attempts to formulate dynamics from relational consistency rather than from independently prescribed field evolution have appeared in several contexts. In absorber formulations of electrodynamics, interaction laws arise from mutual consistency between degrees of freedom rather than from independent field variables \cite{wheeler1945interaction,wheeler1949classical}. Relational approaches to mechanics and quantum theory similarly define evolution through relations between subsystems rather than with respect to an external time parameter \cite{rovelli1991time, rovelli2011forget}. In discrete and lattice gauge formulations, geometry is encoded in local transport relations defined on a prescribed lattice and symmetry group, with curvature determined by consistency around loops \cite{wilson1974confinement, kogut1979introduction}. Analogous ideas also appear in information geometry and geometric learning, where evolution or learning is formulated on representation spaces endowed with specified metrics, connections, or symmetry constraints \cite{amari2000methods, bronstein2017geometric, bruna2014spectral, lecun2022path}. In our previous work on Hebbian Physics Networks (HPN) \cite{auti2026hebbian}, local dynamics was similarly driven by the relaxation of mismatch in conserved quantities, enabling adaptive transport coefficients to emerge from local interaction.

Across these approaches, relational consistency plays a central role, but the structure used to compare neighboring states is typically specified in advance. The manifold, lattice, symmetry group, or representation space is chosen \textit{a priori}, and the dynamical variables evolve within this prescribed comparison framework. While geometric or field variables defined on this structure may adapt, the admissible form of the comparison rule itself remains restricted to a fixed class, so that the transport relations $\mathcal{T}_{ij}$ evolve according to a predefined dependence on the state [Fig.~\ref{fig:figure1}(a)]. As a result, the transport relations are not treated as independent dynamical variables, and the comparison structure does not participate in the dynamical closure.

In contrast to the preceding approaches, we do not assume that the relations used to compare neighboring states are specified in advance. Instead, the transport relations themselves are treated as independent dynamical variables that coevolve with the state, so that the comparison structure becomes part of the dynamical closure.

We formulate dynamics at a level prior to specifying constitutive relations or a closed evolution operator. Rather than assuming a predefined mapping from the state to its time derivative, we adopt an operational description based only on quantities accessible through local interaction. The system is described by state variables representing local densities of globally conserved quantities, whose evolution arises through redistribution between neighboring regions. We retain locality, global conservation, and covariance under independent changes of local reference frame, but do not assume a prescribed transport geometry or a universally synchronized notion of time. Under these conditions, the effective geometry relating neighboring regions is not specified \emph{a priori}, but emerges through the relaxation of local incompatibility, as illustrated in Fig.~\ref{fig:figure1}(b). This extends relational approaches to dynamics \cite{rovelli2011forget} by allowing the comparison structure itself, rather than only the state defined on it, to participate in the evolution.

Within this setting, local incompatibility $r_{ij}$ emerges as the fundamental relational quantity between neighboring regions, defined as the covariant mismatch under the current transport relation $\mathcal{T}_{ij}$ [see Fig.~\ref{fig:figure1}(c) and (d)]. Incompatibility is therefore an interfacial quantity: it vanishes if and only if the two neighboring descriptions can be identified under the instantaneous transport relation, and is nonzero otherwise. In Sec.~\ref{sec:residual} we define node-wise residuals $R_i$ as the accumulation of interfacial incompatibilities at each region, representing unresolved local mismatch and admitting a finite-volume interpretation with a well-defined continuum limit. From these, we construct a scalar functional in Sec.~\ref{sec:L_def} that quantifies instantaneous incompatibility. Dynamics is then defined as the finite-rate relaxation of this global measure $\mathcal L$ through coupled evolution of the state variables and transport relations, as detailed in Sec.~\ref{sec:updates}.

Classical closed evolution laws arise not as primary postulates, but as limiting descriptions of this relaxation process once transport relations stabilize (see Supplementary Information Sec. S2). When transport is prescribed \emph{a priori} and incompatibility is removed instantaneously, evolution reduces to motion on a constraint manifold—for example, in incompressible flow, where pressure enforces instantaneous consistency of the velocity field. More generally, projection methods and elliptic constraint solvers correspond to regimes in which mismatch is eliminated kinematically rather than relaxed dynamically. By contrast, when transport structure evolves and incompatibility persists at finite rate, no prior closure or global time parameter is required, and classical PDEs emerge only as effective coarse-grained limits.

At an operational level, each local region carries a state $U_i$, and neighboring regions are compared through transport relations that map one local frame to another. When the transported neighbor state $\mathcal T_{ij}U_j$ fails to coincide with $U_i$, a local incompatibility $r_{ij}$ arises. Dynamics is defined as the finite-rate relaxation of this incompatibility through the coupled evolution of both the state variables and the transport relations. In this sense, dynamics is not postulated as an evolution law in time, but emerges from the progressive restoration of local consistency.

The remainder of the paper develops the mathematical structure of this framework and derives the coupled state–structure evolution rules implied by this formulation.

\section{Mathematical formulation} 
We now introduce the mathematical framework used in this work. The formulation is developed axiomatically, beginning with the specification of primitive quantities and the minimal structural requirements needed to relate them. No closed evolution law, constitutive relation, transport geometry, or variational principle is assumed at the outset. Instead, the framework identifies which objects must be postulated and which can be constructed from consistency requirements between neighboring descriptions. Subsequent sections show how local comparison, accumulation of mismatch from these comparisons, and coupled relaxation emerge from these axioms and in appropriate limits lead to familiar continuum formulations.

\subsection{Axioms} 
\label{sec:axioms}
The formulation is based on a small set of structural assumptions with clear antecedents in established physical and mathematical theories, but not typically adopted together within a single dynamical description. These include locally defined state variables associated with finite regions, locality of interaction, the existence of a differentiable manifold without prescribed transport geometry, covariance under independent local frame transformations, non-negative measures of mismatch, finite-rate structural adaptation, local smoothness, and the existence of a norm on the state space.

We now state the axioms explicitly.

\setcounter{axm}{0}
\begin{axm}\textit{State is primitive.}
\label{ax:axm0}
The system is described at a given resolution by a set of local state variables \[U = (U_\alpha), \quad \alpha \in \{1,2,3,…\},\] where each 
$U_\alpha$ represents the density of a globally conserved carrier (such as mass, charge, or momentum) associated with a local control region $\Omega$. At the resolved scale, $U_\alpha$ admits no intrinsic local creation or annihilation: its local change arises solely through exchange with neighboring regions and through explicitly represented conversion between tracked carriers. Variables that act only as constraint enforcers, response parameters, or constitutive descriptors (such as pressure or temperature in standard continuum descriptions) are not included in $U$. 
\end{axm}

\begin{axm}\textit{Locality.} 
\label{ax:axm1}
Physical evolution is governed exclusively by local interactions. Each degree of freedom has access only to quantities defined within its immediate neighborhood, and all updates are driven by locally measured imbalance.
\end{axm} 

\begin{axm}
\label{ax:axm2}
\textit{The manifold exists \textit{a priori}; the geometry does not.} 
The system is assumed to admit a smooth differentiable manifold structure, providing local charts, neighborhoods, and tangent spaces sufficient to define locality and continuity. However, no transport geometry is prescribed \emph{a priori}: no background metric, connection, or constitutive relation is assumed.

Instead, the relations governing transport and comparison between neighboring regions are treated as dynamical objects that evolve with the state. Neighboring degrees of freedom are therefore compared only through the currently available transport relations, which adapt in response to inconsistency of description.

Local changes of frame are identified with gauge transformations belonging to a group (typically a Lie group in the smooth limit) acting on the local state space, so that geometric information is encoded in a dynamical connection rather than a fixed background structure.
\end{axm}

\begin{axm}
\label{ax:axm3}
\textit{Local frame freedom (gauge covariance).} 
The local state at each point is defined only up to an arbitrary choice of internal frame. Physical descriptions must therefore be invariant under independent local frame transformations, and any admissible transport relation or measure of incompatibility must transform covariantly under such changes of frame.

Consistency between neighboring descriptions is thus defined relative to the current transport relations connecting their local frames.
\end{axm}

\begin{axm}
\label{ax:axm4} 
\textit{Coercivity and regularity.} 
Local incompatibility is defined as the failure of neighboring descriptions to satisfy the relevant consistency conditions under the current transport relations. One or more non-negative measures may be associated with distinct transport channels. Each such measure vanishes when compatibility is satisfied and grows with the magnitude of unresolved mismatch, implying a positive cost for incompatibility.

Admissible configurations are assumed to remain bounded: both the state variables and the transport relations have finite magnitude. This ensures that incompatibility measures admit well-defined norms and that the resulting dynamics remains regular.
\end{axm}

\begin{axm}
\label{ax:axm5}
\textit{Finite-rate structural response.} 
Transport geometry does not adapt instantaneously, but evolves over a finite internal scale. Transport relations therefore persist across updates and relax through local, generally dissipative adaptation rather than immediate equilibration.

This introduces a separation between state evolution and structural adaptation.
\end{axm}

\begin{axm}
\label{ax:axm6} 
\textit{Local linearizability of transport.}  
Whenever transport relations between neighboring degrees of freedom are well defined, they are assumed to be locally smooth and sufficiently close to the identity to admit a tangent-space (Lie algebra) representation. This permits local linearization of transport operators and the definition of infinitesimal generators.
\end{axm}

\begin{axm}
\label{ax:axm7} 
\textit{Normability.} 
The local state space admits an inner product (or norm) sufficient to define finite incompatibility measures and their magnitude.
\end{axm}

The axioms introduced above are not ad hoc assumptions, but reflect structural elements that appear across several established theoretical frameworks. The use of locally defined state variables associated with finite regions is standard in continuum mechanics and statistical physics \cite{landau1987fluid}. Locality of interaction underlies relativistic field theory, finite-volume methods, and lattice gauge formulations \cite{wilson1974confinement}. The separation between manifold structure and transport geometry is familiar from differential geometry and gauge theory \cite{levicivita2022}. Covariance under independent local frame transformations corresponds to the gauge principle in modern field theory \cite{yang1954conservation,truesdell2004non}. The use of non-negative measures of mismatch to drive relaxation appears in non-equilibrium thermodynamics and gradient-flow formulations \cite{de2013non}. Finite-rate adaptation of structural relations is implicit in kinetic theory, plasticity, and learning dynamics \cite{prigogine1967symmetry,lecun2022path}. Local smoothness and tangent-space representations are standard in geometric transport \cite{levicivita2022,yang1954conservation}, and the use of norms to quantify deviation from compatibility is common in variational and information-geometric frameworks \cite{amari2000methods,bronstein2017geometric}. 

What distinguishes the present formulation is not the introduction of new axioms, but the simultaneous adoption of these familiar axioms without assuming a predefined evolution law, fixed transport geometry, or globally synchronized notion of time. Under these conditions, the admissible form of local incompatibility and the resulting coupled state–structure dynamics are not arbitrary, but restricted by these structural requirements.

Together, these axioms define a minimal framework for systems in which local conservation and transport are well defined, without assuming a closed evolution law or variational principle \emph{a priori}. Within this setting, local incompatibility becomes a physically meaningful quantity whose relaxation drives both state evolution and structural adaptation.

In regimes where compatibility is progressively restored and transport relations stabilize, the resulting dynamics admit effective descriptions consistent with familiar continuum, projection, and variational formulations. Classical evolution laws therefore appear as limiting cases of a more general local relaxation process rather than as primary postulates.

In the following sections, we show that under these axioms, both the form of admissible structural adaptation and the conditions under which classical dynamical descriptions become applicable are constrained.

\subsection{Local symmetry and compatibility}
\label{sec:local_sym_and_compat}

\subsubsection{Covariant difference}
Let $\mathcal{M}$ be a smooth $n$-dimensional manifold representing the configuration or physical space of the system [Axiom~(\ref{ax:axm2})]. We consider a collection of points on $\mathcal{M}$ labeled by indices $i,j: j \in \mathcal{N}(i)$, where $\mathcal{N}(i)$ denotes the neighborhood of point $i$. Each associated with coordinates $\vec{x}_i\in\mathbb{R}^n$ in a chosen external chart. The labels $i$ and $j$ identify locations on the manifold, while the coordinates $\vec{x}_i$ provide a convenient parametrization.

At each point $i$, we associate a local state variable $U_i \equiv U(\vec{x}_i)\in\mathbb{R}^d$ (or $\mathbb{C}^d$), representing the resolved degrees of freedom at that location.

Each point $i$ is further endowed with a local frame represented by an element $g_i\in G$, where $G$ is a Lie group acting (via a chosen representation) on the state space of $U_i$. A local change of frame is defined pointwise by the gauge transformation
\begin{equation}
U_i \mapsto U_i' \equiv g_iU_i,
\label{eq:gauge_transform_U}
\end{equation}
which encodes the fact that the numerical representation of the state depends on a local choice of frame at $i$.

To define \textit{nearby} regions, we introduce a local neighborhood structure on $\mathcal{M}$.  Because no transport geometry or metric is assumed \textit{a priori} [Axiom~(\ref{ax:axm2})], locality cannot be defined through a fixed distance intrinsic to the system. Instead, neighborhood relations are defined in the observer chart used to represent the manifold.

Let $x_i \in \mathbb{R}^n$ denote the coordinates of point $i$ in a chosen external chart. For each point $i$, we define a neighborhood $\mathcal{N}_\epsilon(i)$ consisting of points $j$ whose coordinate separation from $i$ is smaller than a prescribed resolution scale $\epsilon$ according to some observer-defined notion of distance.

Concretely, one may write for example
\[
\|x_j - x_i\|_\infty \le \epsilon,
\]
where $\|\cdot\|_\infty$ denotes the Chebyshev (max-norm) distance in the chosen chart. This choice is made only for convenience and does not play a structural role.

Any definition of neighborhood that induces the same local topology is admissible, including Euclidean distance, $\ell^1$ distance, $k$-nearest-neighbor relations, Delaunay/Voronoi connectivity, or other discretizations consistent with the observer’s resolution. All subsequent constructions depend only on the existence of a finite local neighborhood and on local smoothness, and are therefore invariant under the particular choice of norm or discretization used to define $\mathcal{N}_\epsilon(i)$.

In this sense, locality is defined topologically rather than metrically: the framework requires only that each degree of freedom interacts with a finite set of nearby neighbors, as determined in the observer chart, while the transport geometry relating those neighbors is treated as a dynamical quantity. This reflects a statement about finite information access in the observer chart rather than the existence of a fixed intrinsic metric on the system.

In the present formulation, the neighborhood structure $\mathcal{N}_\epsilon(i)$ is fixed by the observer’s resolution scale and does not evolve dynamically. The dynamical variables act only on the transport operators that relate neighboring points. Allowing the neighborhood itself to change, corresponding to structural adaptation of the interaction graph, lies beyond the scope of the present work and will be considered separately.

To compare states residing in different local frames, we introduce for each ordered pair $(i\leftarrow j)$ with $j\in\mathcal{N}_\epsilon(i)$ a transport operator $\mathcal{T}_{ij}$ that maps quantities defined at point $j$ into the local frame at point $i$. This operator encodes the local transport geometry relating neighboring points on the manifold. 

Gauge consistency requires the transport operator to transform by conjugation under local frame changes,
\begin{equation}
\mathcal{T}_{ij} \mapsto \mathcal{T}_{ij}' \equiv g_i\mathcal{T}_{ij}g_j^{-1}.
\label{eq:gauge_transform_T}
\end{equation}
With this transformation rule, the transported neighbor state transforms covariantly in the frame at $i$:
\begin{equation}
\mathcal{T}_{ij}U_j \mapsto \mathcal{T}_{ij}'U_j'
= (g_i\mathcal{T}_{ij}g_j^{-1})(g_jU_j)
= g_i(\mathcal{T}_{ij}U_j).
\label{eq:transport_covariance}
\end{equation}
Eq.~\eqref{eq:transport_covariance} ensures that differences of the form
\begin{equation*}
\mathcal{T}_{ij}U_j - U_i
\end{equation*}
are well-defined, frame-consistent quantities at point $i$.

\subsubsection{Local smoothness and tangent-space representation}

We now formalize the local smoothness of underlying the transport geometry. Under [Axiom~(\ref{ax:axm6})], whenever transport relations are well defined, each point admits a neighborhood in which the transport operators can be represented within a single local coordinate chart of the Lie group. This permits a tangent-space description of transport relations without assuming global flatness or trivial topology.

\begin{lem}
\label{lem:local_lie_algebra}
Given local smoothness [Axiom~(\ref{ax:axm6})], for each point $i$ and it's neighbor $j \in \mathcal{N}(i)$, there exists a unique element $W_{ij} \in \mathfrak{g}$ such that the transport operator $\mathcal{T}_{ij}$ is given by the exponential map $\mathcal{T}_{ij} = \exp(W_{ij})$. In a local linear representation, this admits the tangent-space approximation  
\begin{equation}
    \mathcal{T}_{ij} \approx I + W_{ij},
    \label{eq:T_equals_I+W}
\end{equation}
which remains gauge-covariant under local transformations. 
\end{lem}
\begin{proof}
Since $G$ is a smooth Lie group, its tangent space at the identity $T_e G$ is canonically identified with the Lie algebra $\mathfrak{g}$. By the normal-neighborhood assumption, the exponential map $\exp: \mathfrak{g} \to G$ is a local diffeomorphism. Thus, for $\mathcal{T}_{ij}$ sufficiently close to the identity, the exponential map provides a local parametrization of $\mathcal T_{ij}$ in terms of Lie algebra element $W_{ij} \in \mathfrak{g}$, enabling the representation $\mathcal{T}_{ij} = \exp(W_{ij})$. A Taylor expansion yields:
\begin{equation}
    \mathcal{T}_{ij} = I + W_{ij} + \mathcal{O}(\|W_{ij}\|^2) 
\end{equation} 
Under a local gauge transformation $(g_i, g_j) \in G \times G$, the operator transforms as $\mathcal{T}_{ij} \mapsto g_i \mathcal{T}_{ij} g_j^{-1}$. For infinitesimal transport where $g_j = g_i + \delta g$, this transformation induces the adjoint action on $W_{ij}$ plus a gradient term, consistent with the transformation of a discrete connection. Consequently, the tangent-space representation is preserved under local gauge transformations to leading order. Continuum limit of this lemma, is discussed in Supplementary Information Sec. S1A.
\end{proof}

The framework itself is not restricted to this truncation: retaining the full nonlinear transport relation $\mathcal{T}_{ij}=\exp(W_{ij})$, or higher-order expansions thereof, preserves gauge consistency and yields well-defined incompatibility measures and adaptation dynamics. Including higher-order terms incorporates finite-transport corrections without altering the conceptual structure developed here.

\subsubsection{Incompatibility}
The assumed symmetry of the description is the freedom to choose independent local frames at each point, represented by gauge transformations of the form $U_i \mapsto g_i U_i$. Two neighboring state descriptions $(U_i, U_j)$ are said to be locally compatible under the current transport geometry if there exists a frame-consistent identification of the state at $j$ with that at $i$. Operationally, this requires that transporting $U_j$ to point $i$ using the current transport relation reproduces $U_i$,
\begin{equation}
U_i = \mathcal{T}_{ij} U_j .
\label{eq:local_compatibility}
\end{equation}

In most geometric learning and information-geometric formulations, structural compatibility is treated as a fundamental constraint of the model space. Neighboring representations are assumed to admit a consistent identification via a prescribed transport, equivariant mapping, or parallel transport operator. This assumption underlies geometric deep learning frameworks in which symmetry and equivariance are enforced by architectural design \cite{bronstein2017geometric}. Similarly, in information-geometric formulations of learning, evolution is described as motion on a smooth, pre-existing model manifold endowed with a fixed geometric structure \cite{amari2000methods}. In these settings, deviations from Eq.~\eqref{eq:local_compatibility} are typically interpreted as approximation error, sampling noise, or transient optimization artifacts, rather than as physically meaningful deviations.

This perspective implicitly assumes instantaneous enforcement of compatibility, equivalent to infinite signal or information propagation speed. While appropriate for static or equilibrium descriptions, this assumption is not justified for physical systems in which information propagates at finite rate. Following a localized perturbation, there exists a nonzero interval ($t \to 0^+$) during which neighboring regions have not yet exchanged sufficient information to restore mutual consistency, and the condition in Eq.~\eqref{eq:local_compatibility} generically fails. Such transient incompatibility is therefore not a modeling deficiency, but a direct consequence of causality and finite-rate transport [Axiom~(\ref{ax:axm5})]. This viewpoint is consistent with non-equilibrium physics, where local symmetry is not enforced instantaneously but is restored dynamically, as seen in pattern-forming systems and active matter \cite{cross1993pattern, ramaswamy2010mechanics, marchetti2013hydrodynamics}.

Accordingly, when Eq.~\eqref{eq:local_compatibility} fails, the neighborhood does not admit a consistent local identification under the current transport relation $\mathcal{T}_{ij}$. In this sense, incompatibility is defined as a \emph{failure of local identification}. Although the underlying symmetry permits a frame-consistent comparison, the presently realized transport geometry does not effect such an identification.

Any admissible measure of this failure must therefore quantify the deviation from Eq.~\eqref{eq:local_compatibility} in a manner that is local and gauge-consistent. We now show that these requirements, when restricted to first-order consistency, identify a minimal covariant form for such an incompatibility measure.

\begin{theorem}
\label{thm:incompatibility}
Let $r_{ij}$ be a local incompatibility measure associated with an ordered pair $(i\leftarrow j)$, constructed from the local states $U_i$, $U_j$ and the transport operator $\mathcal{T}_{ij}$. Suppose that $r_{ij}$ satisfies the following conditions:

\noindent(i)~Locality [Axiom~(\ref{ax:axm1})]: $r_{ij}$ depends only on $(U_i,U_j,\mathcal{T}_{ij})$.

\noindent(ii)~Gauge covariance [Axiom~(\ref{ax:axm3})]: under independent local frame transformations
\[
U_i \mapsto g_i U_i, \quad
U_j \mapsto g_j U_j, \quad
\mathcal{T}_{ij} \mapsto g_i \mathcal{T}_{ij} g_j^{-1},
\]
the incompatibility transforms covariantly at point $i$ (with respect to a $G$-invariant representation),
\[
r_{ij} \mapsto r'_{ij} = g_i r_{ij}.
\]

\noindent(iii)~Compatibility (Eq.~\eqref{eq:local_compatibility}): 
\[
r_{ij}=0 \quad \text{iff} \quad U_i=\mathcal{T}_{ij}U_j.
\]

\noindent(iv)~First order consistency [Axiom~(\ref{ax:axm6})]: $r_{ij}$ reduces linearly to zero as $U_i \to \mathcal{T}_{ij}U_j$.

\vspace{2pt}
\noindent Then, up to an invariant scalar or the application of a $G$-equivariant linear operator, the minimal representation of $r_{ij}$ consistent with these axioms is the covariant difference:
$$
r_{ij} \propto \mathcal{T}_{ij}U_j - U_i.
$$
\end{theorem}

\begin{proof}
By locality, $r_{ij}$ may depend only on $(U_i,U_j,\mathcal{T}_{ij})$. Gauge covariance requires that $r_{ij}$ transform under a local frame change at $i$ in the same way as $U_i$.

Consider the transported neighbor state $\mathcal{T}_{ij}U_j$. By construction,
\[
\mathcal{T}_{ij}U_j \mapsto g_i(\mathcal{T}_{ij}U_j),
\]
so both $U_i$ and $\mathcal{T}_{ij}U_j$ transform identically under gauge transformations at $i$. Any gauge-covariant quantity constructed from $(U_i,U_j,\mathcal{T}_{ij})$ must be expressible in terms of these two objects in the local frame at $i$.

The compatibility condition requires $r_{ij}$ to vanish if and only if $U_i=\mathcal{T}_{ij}U_j$. This excludes any construction that depends on $U_i$ or $\mathcal{T}_{ij}U_j$ separately. Accordingly, $r_{ij}$ must vanish precisely on the diagonal subset $\{(U_i,~\mathcal{T}_{ij}U_j):U_i = \mathcal{T}_{ij}U_j\}$.
Consequently, there exists a (possibly nonlinear) $G$-equivariant map $F$ such that
\begin{equation*}
r_{ij} = F\left(\mathcal{T}_{ij}U_j-U_i\right),
\label{eq:rij_as_equivariant_map}
\end{equation*}
with $F(0)=0$.

[Axiom~(\ref{ax:axm6})] then implies that the compatible configuration admits a well-defined tangent-space description: in a neighborhood where the transport relation is well defined, $F$ is differentiable at the origin and its first non-vanishing  term is linear. Therefore,
\begin{equation*}
r_{ij} = \phi \left(\mathcal{T}_{ij}U_j-U_i\right) + \mathcal{O}(\|\mathcal{T}_{ij}U_j-U_i\|^2),
\label{eq:rij_linearization}
\end{equation*}
where $\phi$ is a $G$-equivariant linear operator. In the absence of additional structural information or higher-order coupling terms, the simplest admissible form—which we identify as the primitive mismatch—is proportional to the covariant difference $\mathcal{T}_{ij}U_j-U_i$. This term represents the lowest-order gauge-covariant approximation of local incompatibility.
\end{proof}

\begin{defn}\textbf{Incompatibility.}
Theorem~\ref{thm:incompatibility} provides the admissible form of a local, gauge-consistent measure of failure of compatibility associated with a single transport relation. We therefore define the oriented local incompatibility associated with the ordered pair $(i \leftarrow j)$ as
\begin{equation}
r_{ij} \equiv \phi \big( \mathcal{T}_{ij} U_j - U_i \big),
\label{eq:rij_definition}
\end{equation}
where $\phi$ is a scalar weight invariant under local gauge transformations.
\end{defn}

The factor $\phi$ may encode material parameters such as conductivity or diffusivity, or discretization-dependent scaling. It introduces no additional frame dependence and no further structural assumptions beyond those required by Theorem~\ref{thm:incompatibility}.

\subsection{Residual}
\label{sec:residual}
Local compatibility is defined pointwise by the condition $U_i=\mathcal{T}_{ij}U_j$. Its violation, however, is not an intrinsic property of an individual point, but a relational property between neighboring regions. Incompatibility $r_{ij}$, therefore, manifests operationally only through comparisons across interfaces separating adjacent control volumes. Physical evolution, however, requires a node-wise quantity that summarizes the unresolved mismatch experienced at a point. For this reason, we introduce the \textit{residual} $R_i$ as a local accumulation of interfacial incompatibilities incident to node $i$. By construction, $R_i$ must depend only on neighbors of $i$ and remain finite under refinement of the neighborhood, so that it represents a well-defined local measure of unresolved incompatibility and admits a continuum limit.

\begin{cor}
\label{cor:node_residual}
Let $r_{ij}$ denote the oriented incompatibility associated with the pair $(i \leftarrow j)$. Any admissible node-wise residual $R_i$ that represents a local accumulation of these incompatibilities and remains finite under refinement of the neighborhood admits the following canonical representation:
\begin{equation}
R_i = \frac{1}{\Omega_i}\sum_{j\in\mathcal{N}(i)} A_{ij} r_{ij},
\end{equation}
up to an invariant scalar or $G$-equivariant linear operator. Here $\Omega_i$ is the local control volume and $A_{ij}$ is the measure of the interface.
\end{cor}

\begin{proof}
By Theorem~\ref{thm:incompatibility}, the primitive local incompatibility associated with an ordered pair $(i\leftarrow j)$ is defined, up to invariant linear maps, by the covariant mismatch
$r_{ij}=\phi(\mathcal{T}_{ij}U_j-U_i)$.
Each such incompatibility is an interfacial quantity that transforms covariantly in the local frame at node $i$.

By locality, any node-wise residual $R_i$ may depend only on incompatibilities incident to node $i$, and therefore only on the collection
$\{r_{ij}\}_{j\in\mathcal{N}(i)}$.
Gauge covariance then restricts $R_i$, to leading order, to a linear combination of these quantities with invariant scalar weights,
\[
R_i = \sum_{j\in\mathcal{N}(i)} a_{ij} r_{ij}.
\]

The coefficients $a_{ij}$ are constrained by the requirement that $R_i$ remain finite under refinement of the neighborhood.
As the discretization is refined, the number of neighbors increases while the interface measures $A_{ij}$ and control volume $\Omega_i$ scale with resolution. To ensure that $R_i$ behaves as a well-defined density (i.e., a divergence-like operator), the weights must scale with the ratio of interface measure to volume.\\
Consequently, the minimal consistent construction that preserves the intensive nature of the residual is the area-weighted average over the control volume,
\[
R_i = \frac{1}{\Omega_i}\sum_{j\in\mathcal{N}(i)} A_{ij} r_{ij}.
\]
\end{proof}

\subsection{Global incompatibility measure}
\label{sec:L_def}


Having defined local incompatibility and its accumulation into node-wise residuals, we now introduce a global measure of unresolved mismatch. Because interactions are local and updates are asynchronous, there is no preferred global ordering parameter with respect to which evolution can be compared directly. Any scalar diagnostic of the system must therefore be constructed from locally defined quantities and remain invariant under independent local frame transformations.

The natural building blocks for such a measure are the node-wise residuals $R_i$, which represent the canonical local accumulation of interfacial incompatibility. A global measure must therefore aggregate these contributions additively, without introducing non-local couplings or external structure. Under these requirements, the admissible form of a global incompatibility measure is constrained.

\paragraph*{Structure of the global incompatibility measure near compatibility.}
Let $R_i$ denote the node-wise residual and $W_{ij}$ the edge-wise transport parameters. We seek a scalar functional $\mathcal{L}[U,W]$ that quantifies instantaneous incompatibility. While many invariant functionals are possible, their admissible form is constrained by the axioms.

\emph{Locality} [Axiom~(\ref{ax:axm1})] suggests an additive decomposition into node-wise and edge-wise contributions,
\begin{equation}
\mathcal{L} = \sum_i \ell(R_i) + \sum_{(i,j)} q(W_{ij}),
\end{equation}
precluding non-local couplings. \emph{Gauge invariance} [Axiom~(\ref{ax:axm3})] and \emph{Coercivity} [Axiom~(\ref{ax:axm4})] require that $\ell$ and $q$ be invariant under their respective representations and vanish only at the compatible state ($R=0, W=0$).

Near compatibility, local smoothness [Axiom~(\ref{ax:axm6})] and normability [Axiom~(\ref{ax:axm7})] imply that any admissible functional is captured at leading order by its quadratic expansion. In this regime, the canonical forms for $\ell$ and $q$ are:
\begin{align}
\ell(R)&\approx \tfrac12 \langle R,R\rangle, \nonumber \\
q(W)&\approx \tfrac12 \langle W,W\rangle,
\end{align}
where the inner products are $G$- and $\mathrm{Ad}$-invariant, respectively.

Because $R$ and $W$ transform in inequivalent representations, no gauge-invariant linear coupling exists between them at leading order. Consequently, the global measure is block-diagonal in the residual and structural variables. Collecting these terms, we identify the simplest admissible quadratic form for the global incompatibility measure:
\begin{equation}
\mathcal{L}[U,W] = \frac12 \sum_i \langle R_i,R_i\rangle+\frac{\lambda}{2} \sum_{(i,j)} \langle W_{ij},W_{ij}\rangle,
\end{equation}
where $\lambda>0$ acts as a structural stiffness parameter. This characterization provides a robust energy-like functional for driving relaxation dynamics in the vicinity of the compatible state.

\paragraph*{Interpretation of the structural stiffness term.}
The quadratic term in $W_{ij}$ assigns a finite cost to the transport structure itself. Because the transport relations encode the geometry through which neighboring states are compared, this term sets a scale for structural adaptation: small $\lambda$ permits rapid reconfiguration, while large $\lambda$ favors persistence of the existing geometry. This term may be viewed as a minimal covariant penalty on the connection, analogous to mass terms in gauge theories \cite{proca1936theorie, hehl2017gauge}, but here it reflects finite responsiveness rather than symmetry breaking.

The functional $\mathcal{L}[U,W]$ provides a gauge-invariant scalar measure of instantaneous incompatibility, combining both residual mismatch and the magnitude of the transport structure. Near compatibility, it defines a quadratic norm on admissible state–structure configurations. In regimes where the transport geometry is approximately fixed, $\mathcal{L}$ reduces to functionals familiar from gradient-flow and non-equilibrium thermodynamics. More generally, however, it extends these descriptions by including the geometry itself as a dynamical variable. The role of $\mathcal{L}$ is therefore not to prescribe dynamics, but to characterize the instantaneous configuration of the system. The evolution arises from its local relaxation through coupled updates of $U$ and $W$, rather than descent on a fixed functional defined over a prescribed geometry.

\subsection{Coupled relaxation of state and structure}
\label{sec:updates}
Having constructed the gauge-invariant incompatibility measure $\mathcal{L}[U,W]$, the framework admits a natural dynamical interpretation in which evolution proceeds through the relaxation of unresolved incompatibility. Because $\mathcal{L}$ depends on both the state variables $U$ and the transport structure $W$, this relaxation induces a coupled evolution of state and structure.

A canonical realization of this principle is to take variations of $\mathcal{L}$ with respect to $U$ and $W$ as the drivers of evolution, leading to gradient-driven update rules derived in the following section.

Although expressed in terms of a global functional, $\mathcal{L}$ is constructed as a sum of local contributions depending only on neighboring states and transport relations. Consequently, its variation yields update rules that depend only on local information. In this sense, $\mathcal{L}$ provides a global representation of locally defined incompatibility rather than a nonlocal objective requiring global coordination.

\subsubsection{Variation with respect to $W_{ij}$: structure relaxation} 
We first compute the gradient of $\mathcal{L}$ with respect to an individual edge connection $W_{ij}$, holding all other variables fixed. Under the directed convention adopted throughout this work, the interfacial incompatibility $r_{ij}$ is expressed in the local frame at node $i$.
Consequently, the transport parameter $W_{ij}$ appears \emph{only} in the residual $R_i$ and does not enter $R_m$ for any $m\neq i$.

Using the linearized transport representation $\mathcal{T}_{ij} \approx I+W_{ij}$, the node-wise residual $R_i$ is approximated as: 
\begin{equation}
R_i = \sum_{j\in\mathcal{N}(i)}\frac{\phi A_{ij}}{\Omega_i}\big(U_j-U_i + W_{ij}U_j\big),
\label{eq:Ri_linearized}
\end{equation}
and the variation of $R_i$ with respect to $W_{ij}$ is
\begin{equation}
\delta R_i = \frac{\phi A_{ij}}{\Omega_i}~~\delta W_{ij}~~U_j .
\label{eq:dRi_dWij}
\end{equation}
Let $\langle\cdot,\cdot\rangle$ denote the Euclidean inner product on the state space, and let $\langle A,B\rangle_F=\mathrm{tr}(A^\top B)$ denote the Frobenius inner product on matrices. We consider the first variation of $\mathcal{L}$ under the infinitesimal change $W_{ij} \mapsto W_{ij} + \epsilon \delta W_{ij}$, where $\epsilon \to 0$. The variation is given by:
\begin{align}
\delta\mathcal{L} 
&= \delta\left(\frac12\sum_i \|R_i\|^2\right) + \delta\left(\frac{\lambda}{2}\sum_{(i,j)\in E}\|W_{ij}\|_F^2\right) \nonumber\\
&= \sum_i \langle R_i,~~\delta R_i\rangle + \lambda\sum_{(i,j)\in E}\langle W_{ij},~~\delta W_{ij}\rangle_F \nonumber\\
&= \langle R_i,~~\delta R_i\rangle + \lambda\langle W_{ij},~~\delta W_{ij}\rangle_F,
\label{eq:first_variation_W}
\end{align}
since $W_{ij}$ uniquely influences $R_i$ and the corresponding $(i,j)$ term in the regularization.

Substituting Eq.~\eqref{eq:dRi_dWij} into the first term:
\begin{align}
\langle R_i,~~\delta R_i\rangle &= \frac{\phi A_{ij}}{\Omega_i}\langle R_i,~~\delta W_{ij}U_j\rangle \nonumber \\ 
&= \frac{\phi A_{ij}}{\Omega_i}\langle R_i U_j^\top,~~\delta W_{ij}\rangle_F,
\label{eq:trace_identity}
\end{align}
where the second line follows from the identity $\langle \mathbf{a}, \mathbf{M}\mathbf{b} \rangle = \text{tr}(\mathbf{a}^\top \mathbf{M} \mathbf{b}) = \text{tr}(\mathbf{b} \mathbf{a}^\top \mathbf{M}) = \langle \mathbf{a}\mathbf{b}^\top, \mathbf{M} \rangle_F$. Combining this with Eq.~\eqref{eq:first_variation_W}  yields:
\begin{equation}
\delta\mathcal{L} = \left\langle \Bigl[\frac{\phi A_{ij}}{\Omega_i}R_i U_j^\top + \lambda W_{ij}\Bigr],~~\delta W_{ij}\right\rangle_F.
\label{eq:variation_final}
\end{equation}
The gradient of $\mathcal{L}$ with respect to $W_{ij}$ is therefore:
\begin{equation} 
\nabla_{W_{ij}}\mathcal{L} = \frac{\phi A_{ij}}{\Omega_i}R_i U_j^\top + \lambda W_{ij}.
\label{eq:grad_Wij}
\end{equation}
A steepest-descent update with step size $\eta>0$ defines the canonical structural adaptation rule:
\begin{equation}
\Delta W_{ij} = -\eta_w\Big(\frac{\phi A_{ij}}{\Omega_i}R_i U_j^\top + \lambda W_{ij}\Big).
\label{eq:Hebbian_update_node-wise}
\end{equation}
Eq.~\eqref{eq:Hebbian_update_node-wise} defines a strictly local structural adaptation rule. The update of the transport parameter $W_{ij}$ depends only on the residual measured at node $i$, the activity of the neighboring node $j$, and a linear decay term. In the scalar case ($d=1$), this reduces to
\begin{equation}
\Delta W_{ij} = -\eta_w\big(\frac{\phi A_{ij}}{\Omega_i}\underbrace{R_i}_{\rm error}                                \times \underbrace{U_j}_{\rm activity} +
                            \underbrace{\lambda W_{ij}}_{\rm decay}\big),
\end{equation}
showing that the derived structural update has the form of a $(\text{local error} \times \text{activity} + \text{decay})$ rule. This algebraic structure is similar to update rules that appear in Hebbian-type learning models \cite{amari1991mathematical, gerstner2018eligibility}, but here the residual $R_i$ is not a supervised error signal. It is a direct measure of local geometric incompatibility, and the update follows from the minimization of the incompatibility functional under the locality and conservation constraints, rather than from an assumed learning objective.

\subsubsection{Variation with respect to $U_i$: state relaxation}
We compute the variation of $\mathcal{L}$ with respect to the state variables $U_i$, holding the transport structure $W_{ij}$ fixed. Unlike the structural parameters, the state $U_i$ enters the global incompatibility measure through two distinct channels: One, node $i$ as a Destination: $U_i$ appears directly in the local residual $R_i$, where it must be compatible with the incoming transport from neighbors $j \in \mathcal{N}(i)$, and two, when node $i$ acts as a Source: $U_i$ acts as a reference state for all neighboring nodes $k$ that point to $i$ (i.e., $\{k : i \in \mathcal{N}(k)\}$). Any change in $U_i$ therefore perturbs the residuals $R_k$ at those neighboring sites.

Following the definition of $R_i$ from Eq.~\eqref{eq:Ri_linearized}, the variation of the local residual when $i$ acts as a destination is:
\begin{equation}
\delta R_i \big|_\text{dst} = -\sum_{j\in\mathcal N(i)}\frac{\phi A_{ij} }{\Omega_i} \delta U_i = -\frac{\phi A_i}{\Omega_i} \delta U_i,\label{eq:dRi_dUi_dest}
\end{equation}
where $A_i = \sum_{j \in \mathcal{N}(i)} A_{ij}$ is the total interface area through which $i$ receives information.

Conversely, when $i$ acts as a source for a neighbor $j$, its state $U_i$ is transported via $(I + W_{ji})$ to the neighbor's frame. The resulting variation in the neighbor's residual is:
\begin{equation}
\delta R_j \big|_\text{src} = \frac{\phi A_{ji}}{\Omega_j}(I + W_{ji}) \delta U_i.
\label{eq:dRj_dUi_source}
\end{equation}

We consider the infinitesimal transformation $U_i \mapsto U_i + \epsilon \delta U_i$, where $\delta U_i$ is an arbitrary perturbation vector and $\epsilon$ is a vanishingly small scalar parameter. Using the inner product identity $\langle \mathbf{a}, \mathbf{M}\mathbf{b} \rangle = \langle \mathbf{M}^\top \mathbf{a}, \mathbf{b} \rangle$, we collect the variations from both roles to identify the total change in $\mathcal{L}$:
\begin{align}
\delta\mathcal{L} &= \underbrace{\langle R_i, \delta R_i \big|_\text{dst} \rangle}_\text{Destination term} + \underbrace{\sum_{j:i \in \mathcal{N}(j)} \langle R_j, \delta R_j \big|_\text{src} \rangle}_\text{Source term} \nonumber\\
&= \left\langle \Bigl[-\frac{\phi A_i}{\Omega_i} R_i + \sum_{j:i \in \mathcal{N}(j)} \frac{\phi A_{ji}}{\Omega_j} (I+W_{ji})^\top R_j\Bigr], ~~\delta U_i \right\rangle,
\label{eq:variation_Ui_dual}
\end{align}
where the transition to the final line utilizes the linearity of the inner product to factor out the common perturbation $\delta U_i$. By moving the transport operator $(I + W_{ji})$ from the right side of the inner product to the left, it is replaced by its adjoint (transpose) $(I + W_{ji})^\top$. The term inside the inner product with $\delta U_i$ is, by definition, the gradient $\nabla_{U_i}\mathcal{L}$. This gradient consists of a local restoration term (where $i$ acts as a destination) and a neighboring back-propagation term (where $i$ acts as a source):
\begin{equation}
\nabla_{U_i}\mathcal{L} = \underbrace{-\frac{\phi A_i}{\Omega_i} R_i}_\text{Local restoration} + \underbrace{\sum_{j:i \in \mathcal{N}(j)} \frac{\phi A_{ji}}{\Omega_j} (I+W_{ji})^\top R_j}_\text{Neighboring back-propagation}.
\label{eq:grad_Ui_dual}
\end{equation}
A steepest-descent update with step size $\eta_u > 0$ yields:
\begin{equation}
\Delta U_i = \eta_u \left( \frac{\phi A_i}{\Omega_i} R_i - \sum_{j:i \in \mathcal{N}(j)} \frac{\phi A_{ji}}{\Omega_j} (I+W_{ji})^\top R_j \right).
\label{eq:Ui_update_final}
\end{equation}
This update rule demonstrates that state relaxation is not merely a local correction. It is a bidirectional coordination: the first term relaxes $U_i$ to satisfy its own incoming transport, while the second term (the sum over neighbors) adjusts $U_i$ to reduce the incompatibility it ``causes" in its neighbors.

\subsection{Geometric Interpretation of Incompatibility}
The incompatibility defined above can be interpreted geometrically as a curvature associated with the transport geometry on the interaction graph. Nonzero incompatibility implies that transport around closed loops fails to return the state to itself, providing a discrete notion of holonomy. This geometric viewpoint becomes essential in the limiting cases discussed in Sec.~\ref{sec:limiting_cases}, where solenoidal and vorticity-like fields arise from the non-integrability of the transport relations. For completeness, the operational definition of holonomy and plaquette curvature is given in Appendix~\ref{sec:holonomy_curvature}.

\section{Limiting cases}
\label{sec:limiting_cases}
The formulation developed above is intrinsically asynchronous. Each node evolves its state and transport parameters through local residual-driven updates. Updates are indexed only by local relaxation events, which may occur at different rates across nodes and edges and need not admit a global ordering. In this sense, the dynamics is defined entirely in terms of local response to incompatibility, independent of an external clock.

In many physical regimes, however, an external observer does not resolve individual relaxation events. Instead, the system undergoes a large number of  asynchronous updates on a scale much shorter than the observer’s sampling interval. When this separation of scales exists, the detailed ordering of local events becomes unresolvable, and the cumulative effect characterizes the perceived evolution.

In this regime, we may introduce a global bookkeeping parameter $t$, or observer time, to index successive coarse-grained observations. This parameter is an effective coordinate rather than an intrinsic clock; it merely labels the accumulation of many unresolved local updates. This is the standard coarse-graining implicit in continuum physics: whenever the observation scale is large compared to the microscopic event rate, asynchronous dynamics admits an effective description in terms of smooth evolution.

Concretely, if $\Delta_n U_i$ denotes the increment of $U_i$ produced by the
$n$-th local update event affecting node $i$, then in the limit of frequent, small updates one may identify
\begin{equation}
\partial_t U_i \equiv \lim_{\Delta t \to 0} \frac{1}{\Delta t}
\mathbb{E}\Big[\sum_{n\in \mathfrak{E}_i(t,t+\Delta t)} \Delta_n U_i\Big],
\label{eq:time_limit}
\end{equation}
where $\mathfrak{E}_i(t,t+\Delta t)$ denotes the set of update events involving node
$i$ within the observation window $[t,t+\Delta t]$, and $\mathbb E[\cdot]$ denotes an average over local update events occurring within the observation window.

For small $\Delta t$, the expected number of local relaxation events occurring within this interval satisfies
\begin{equation}
\mathbb{E}[n(\Delta t)]\approx\alpha_1\Delta t.
\end{equation}
Assuming that the per-event increment varies negligibly over $[t,t+\Delta t]$, so that $\mathbb{E}[\Delta_n U_i]\approx \Delta U_i$, Eq.~\eqref{eq:time_limit} yields effective evolution equation: 
\begin{equation}
\partial_t U_i = \alpha_1\Delta U_i.
\label{eq:time_limit_2}
\end{equation}

This identification does not assert microscopic synchrony, nor does it require a well-defined intrinsic time variable. It assumes only that the ordering of sufficiently many small local updates becomes immaterial at the coarse-grained level, so that their net effect can be represented by a deterministic flow. 

In the supplementary information, we show detailed derivations of how hyperdiffusion (Sec. S2A), parabolic diffusion (Sec. S2B), incompressible Navier-Stokes (Sec. S2C), and Ampere-Maxwell law of electromagnetism (Sec. S2D) form the limiting cases of the same underlying asynchronous residual relaxation process.

\section{Discussion}
\subsection{Description consistency}
The global functional $\mathcal{L}[U,W]$ should not be interpreted as a free energy, entropy, loss function, or Lyapunov potential in the conventional sense. It is not introduced as a thermodynamic state function, nor as an externally prescribed objective to be optimized. Rather, it quantifies the instantaneous internal inconsistency of the current description: how far the local states and transport relations are from being mutually compatible under the assumed conservation and covariance structure.

The first term measures the accumulated residual mismatch across the system, while the second assigns a finite cost to maintaining nontrivial transport structure. In this sense, $\mathcal{L}$ is a measure of description consistency rather than an energetic quantity. Its decrease under the local update rules reflects the progressive restoration of mutual consistency between neighboring regions, not the minimization of a thermodynamic potential.

This interpretation is particularly transparent in the one-dimensional scalar case. There, the oriented incompatibility $r_{ij}$ plays a role analogous to a local phase mismatch $(\Delta \theta)$: it measures how far two neighboring descriptions fail to align under the current transport relation. From this viewpoint, $\mathcal{L}$ aggregates these local mismatches into a global measure of how internally coordinated the system's present description is.

\subsection{Relation to gauge-theoretic formulations}
The present formulation is closely related to gauge-theoretic descriptions, but differs in where structure is introduced. In conventional approaches, one specifies a symmetry group and an admissible class of connections, from which an action functional is constructed. Conservation laws and dynamical equations then follow from this prescribed structure, for example through Noether’s theorem. In this setting, symmetry is an input at the level of representation, and dynamics evolves within a fixed geometric framework.

In contrast, the present framework does not assume an underlying action or a prescribed class of transport relations. Instead, it begins with locally conserved quantities and independent local frame freedom, and defines incompatibility as the failure of neighboring states to be consistently related under the current transport geometry. The transport relations are then determined through the relaxation of this incompatibility. Symmetry is therefore not imposed as a prior constraint, but appears as a consistency condition satisfied by transport relations that successfully coordinate neighboring regions.

This distinction is not a reversal of Noether’s result, but a shift in the level at which assumptions are introduced. Rather than specifying symmetry and deriving conservation, the framework treats conservation as primitive and derives admissible transport structure from the requirement of local compatibility. In this sense, gauge covariance is assumed at the level of local frame freedom, but the transport structure that realizes consistent comparison is not prescribed in advance and must emerge through local relaxation.

This shift has practical consequences for systems in which the appropriate transport relations or symmetries are not known \emph{a priori}. By not requiring a predefined geometric or variational structure, the framework provides a basis for modeling non-equilibrium and multiscale systems in which the effective interactions and organizing principles must themselves be learned or adapt through local dynamics.

\subsection{Relational time and asynchrony}

A defining feature of this formulation is the absence of an intrinsic global time variable. State variables, transport relations, and residuals are defined locally, and the dynamics is specified not as evolution with respect to a universal clock, but as an ordering of local relaxation events. Any parameter used to index these updates serves only as a bookkeeping device; it does not imply \emph{a priori} simultaneity across the system.

Because the update laws depend exclusively on locally available information, the formulation is intrinsically asynchronous. No global synchronization or universal ``now'' is required for the dynamics to be well defined. In this sense, asynchrony is not a numerical convenience or implementation detail, but the natural mode of evolution for systems governed by finite-rate local interaction.

From this perspective, the progression of the system is determined by the gradual reduction of incompatibility under the currently available transport relations. Classical time-dependent evolution equations arise only as coarse-grained descriptions (Sec.~\ref{sec:limiting_cases}), when many unresolved local updates can be summarized by a smooth global time coordinate. Synchronous evolution then appears as an effective approximation, not as a primitive assumption.

\subsection{Adaptive plasticity and the emergence of structure}
The central result of this axiomatic construction is the identification of a gauge-invariant incompatibility measure $\mathcal{L}[U,W]$ whose relaxation drives the system's evolution. A key consequence of this dynamics is the update rule for the transport structure [Eq. \eqref{eq:Hebbian_update_node-wise}], which operates in two distinct qualitative regimes dictated by the state of the system.

Far from compatibility, the update is dominated by the correlation term $R_i U_j^\top$. In this regime, the transport geometry adapts to align with unresolved residuals, effectively ``learning" the pathways necessary to resolve incompatibility. As the system approaches a compatible state ($R_i \to 0$), the decay term $\lambda W_{ij}$ becomes the primary driver, suppressing the transport structure toward a locally trivial configuration.

This transition from structure-building to structure-suppression is a direct consequence of the coercivity and smoothness of the incompatibility measure. It defines an operational form of adaptive plasticity: the system modifies its internal relations only in response to active mismatch. Once compatibility is restored, unnecessary geometric structure is suppressed. This behavior emerges naturally from the requirement of local covariant relaxation, without the need for external objectives, rewards, or supervised labels.

\vspace{10pt}
To summarize, the present framework shifts the point at which structure enters the description of dynamics. Rather than prescribing differential operators, transport relations, or variational principles \emph{a priori}, it treats locally conserved quantities and local frame freedom as primitive, and identifies the relaxation of incompatibility as the fundamental driver of evolution. In this setting, both the state and the transport structure coevolve to restore local consistency, while familiar continuum descriptions and classical field equations emerge only as effective limiting forms once the transport relations stabilize.

This perspective also changes the role of time and synchronization. Evolution is defined through an ordering of local relaxation events, and the global time used in continuum descriptions appears only as a coarse-grained summary of many unresolved local updates. By not assuming a fixed comparison structure from the outset, the framework provides a basis for describing systems in which the effective relations governing transport and coordination must themselves emerge through local interaction.

\section{Limitations and Outlook}
\subsection{Limitations}
The framework developed in this work is deliberately local and perturbative with respect to the transport geometry. Our structural results rely on the existence of a locally smooth identification between neighboring degrees of freedom through transport operators lying in a normal neighborhood of the identity. This justifies the linearized representation $\mathcal{T}_{ij} \approx I + W_{ij}$, which underpins the quadratic form of the global measure $\mathcal{L}$ and the resulting gradient-driven relaxation.

However, two primary limitations define the boundary of the current formulation:

\begin{enumerate}
    \item Topological Persistence: A limitation of the present formulation is that the neighborhood structure is not treated as a fully dynamical variable. While the transport parameters ($W_{ij}$) evolve and may effectively suppress or deactivate interactions, the creation of new connections is not modeled. In this sense, the framework permits adaptive reweighting and pruning of existing edges, but does not yet include mechanisms for topological growth or rewiring. Extending the formulation to allow the existence of edges themselves to emerge or disappear dynamically based on the accumulated residual stress remains an important direction for future work.
    \item Perturbative Regime: The framework characterizes dynamics in regimes where compatibility is locally attainable through continuous adaptation. Systems characterized by singular transitions, discontinuous rewiring, or genuinely non-perturbative geometric changes—where incompatibility cannot be expressed as a covariant difference in a shared tangent space—fall outside the scope of this formulation.
\end{enumerate}
Rather than attempting to describe the origin of locality from a state of total chaos, this work isolates the regime in which geometry exists and adapts continuously. By focusing on this domain, we characterize how dynamics can be meaningfully expressed as the restoration of local compatibility.

\subsection{Outlook}
Within its stated axioms, the framework provides a minimal constructive description of local dynamics. Incompatibility, its accumulation, and the coupled relaxation of state and transport structure follow directly from locality and covariance. In this view, variational principles, and thermodynamic descriptions are not primitive postulates of the theory; they appear as stable limiting summaries once the transport geometry has sufficiently adapted to resolve local mismatch.

A key implication of this work is that it offers a bottom-up, relational language for ``world modeling," a central challenge in modern machine learning and robotics. Currently, much of the world-model program is top-down: one posits a parameterized predictor and trains it to match observations, with the ``model" emerging implicitly within the fitted function class \cite{lecun2022path}. While influential proposals emphasize learning predictive latent representations \cite{lecun2023}, these paradigms typically rely on an externally fixed representational frame. The ``rules of interaction" are learned from the perspective of that frame, rather than as a structure that must remain consistent under independent local reparameterizations.

The present framework suggests a complementary relational layer: consistency across local frames as a primitive requirement. In physical systems, comparisons are meaningful only through the transport structure $\mathcal{T}_{ij}$ that relates neighboring descriptions. Here, the world model is not a monolithic predictor, but a distributed, gauge-covariant geometry that is continuously negotiated by local incompatibility. This shifts the focus from global prediction to the restoration of mutual consistency, where the state $U$ and the relational structure $W$ co-adapt through purely local signals.

This perspective may also clarify ongoing discussions regarding open-ended objectives for adaptive agents. While imperatives like ``autonomy" or ``exploration" are often invoked, they frequently lack a local, physically composable metric. In contrast, the incompatibility measure $\mathcal{L}$ provides a measurable primitive: an agent or physical substrate is driven by locally witnessed mismatch, triggering structural plasticity whenever that mismatch persists. This provides a concrete mechanism for open-ended adaptation without prescribing a terminal task loss: The system can continually restructure its effective connectivity and constraints as long as incompatibility remains unresolved. Classical ``closed" laws of behavior emerge only when the system has reached a high degree of internal compatibility. Related critiques of current generative modeling emphasize the gap between token-level prediction and robust causal/world understanding, underscoring the need for objectives grounded in interaction and invariance rather than purely in reconstruction or imitation \cite{mitchell2023ai, huang2023}.

Finally, the same mechanism naturally connects to adaptive biological and neural systems, where effective connectivity evolves in response to persistent mismatch, and longer-term structural processes constrain admissible interaction geometries rather than prescribing detailed dynamics. In that reading, ``development" and ``learning" correspond to slow evolution of transport structure (or the neural pathways), while ``behavior" corresponds to fast state relaxation. Both are driven by the same incompatibility signals, but operate at different internal rates.

In conclusion, this work identifies compatibility restoration—expressed as covariant mismatch minimization with finite-rate structural adaptation—as an organizing principle for dynamics. By treating transport geometry as a dynamical object rather than fixed background we provide a constructive route from relational inconsistency to effective classical laws. This bottom-up logic suggests that for artificial systems to build robust world models, the missing ingredient may not be a more powerful global predictor, but a locally enforced, frame-consistent relational geometry that evolves in response to persistent mismatch.

\begin{acknowledgments}
This work was partly supported by JSPS KAKENHI Grant Number JP23K28154 and JST CREST Grant Number JPMJCR24R2.
\end{acknowledgments}

\appendix
\section{Holonomy and plaquette curvature}
\label{sec:holonomy_curvature}
The coupled state--structure update laws derived above depend only on local edge-wise transport operators $\mathcal{T}_{ij}$ and node-wise incompatibilities. To characterize circulation, curl-like effects, and geometric frustration of the learned transport structure, it is necessary to introduce loop-based diagnostics. These quantities do not modify the update laws themselves but provide a canonical geometric interpretation of the resulting transport field.

\subsection{Holonomy along closed loops}
Curvature is a relational property defined by transport around closed paths. For any oriented loop
\begin{equation}
\gamma = (i_0 \rightarrow i_1 \rightarrow \cdots \rightarrow i_n),
\qquad i_n=i_0,
\end{equation}
the holonomy based at node $i_0$ is defined as the ordered product
\begin{equation}
\mathcal{H}_\gamma(i_0)
=
\mathcal{T}_{i_0 i_1}
\mathcal{T}_{i_1 i_2}
\cdots
\mathcal{T}_{i_{n-1} i_n}.
\end{equation}
If the transport structure is locally integrable along $\gamma$, parallel transport returns a state to itself and $\mathcal{H}_\gamma(i_0)=I$. Any deviation from the identity indicates geometric frustration or curvature.

\subsection{Plaquette holonomy}
On a mesh with faces, we take the minimal loops to be the oriented boundaries of plaquettes. For a plaquette $f$ with ordered boundary
\begin{equation}
\partial f = (i \rightarrow j \rightarrow k \rightarrow \ell \rightarrow i),
\end{equation}
the plaquette holonomy is
\begin{equation}
\mathcal{H}_f(i) = \mathcal{T}_{ij} \mathcal{T}_{jk} \mathcal{T}_{k\ell} \mathcal{T}_{\ell i}.
\end{equation}
On general graphs without explicit faces, minimal cycles (e.g. triangles) may be used in an analogous manner.

\subsection{Curvature as a Lie-algebra element}
We define the discrete curvature associated with plaquette $f$ by the logarithm of its holonomy,
\begin{equation}
\kappa_f(i) \equiv \log\big(\mathcal{H}_f(i)\big) \in \mathfrak{g}.
\end{equation}
Vanishing curvature, $\kappa_f(i)=0$, is equivalent to trivial holonomy $\mathcal{H}_f(i)=I$ and characterizes locally integrable transport geometry. Under local changes of frame, $\kappa_f(i)$ transforms covariantly in the adjoint representation.

\subsection{Non-Abelian structure}
Substituting $\mathcal{T}_{ij}=\exp(W_{ij})$, the plaquette holonomy becomes a product of exponentials,
\begin{equation}
\mathcal{H}_f(i) = \prod_{(a\rightarrow b)\in\partial f} \exp(W_{ab}),
\end{equation}
and its logarithm generates nonlinear commutator corrections via the Baker--Campbell--Hausdorff expansion,
\begin{equation}
\kappa_f(i) = \sum_{(a\rightarrow b)\in\partial f} W_{ab} + \frac{1}{2} \sum_{e<e'}[W_e,W_{e'}] + \mathcal{O}(W^3).
\end{equation}
These commutator terms encode non-Abelian curvature and are retained throughout the general formulation.

\subsection{Node-wise curvature diagnostic}
For later use, a node-wise curvature (curl) diagnostic may be defined by averaging the plaquette curvatures incident to node $i$,
\begin{equation}
\kappa_i = \frac{1}{\sum_{f\ni i} A_f} \sum_{f\ni i} A_f\kappa_f(i),
\end{equation}
where $A_f$ denotes the plaquette area (or unity for unweighted graphs).

\bibliography{bibliography}

@article{einstein1935can,
  title = {Can Quantum-Mechanical Description of Physical Reality Be Considered Complete?},
  author = {Einstein, A. and Podolsky, B. and Rosen, N.},
  journal = {Phys. Rev.},
  volume = {47},
  issue = {10},
  pages = {777--780},
  numpages = {0},
  year = {1935},
  month = {May},
  publisher = {American Physical Society},
  doi = {10.1103/PhysRev.47.777},
  url = {https://link.aps.org/doi/10.1103/PhysRev.47.777}
}

@article{maxwell1865viii,
  title={{VIII. A dynamical theory of the electromagnetic field}},
  author={Maxwell, James Clerk},
  journal={Phil. Trans. R. Soc. Lond.},
  number={155},
  pages={459--512},
  year={1865},
  publisher={The Royal Society London},
  url={https://doi.org/10.1098/rstl.1865.0008}
}

@book{goldstein1950classical,
  title={{Classical mechanics}},
  author={Goldstein, Herbert and Poole, Charles P and Safko, John},
  volume={2},
  year={1950},
  publisher={Addison-wesley Reading, MA}
}

@book{poincare1899methodes,
  title={Les m{\'e}thodes nouvelles de la m{\'e}canique c{\'e}leste},
  author={Poincar{\'e}, Henri},
  volume={3},
  year={1899},
  publisher={Gauthier-Villars}
}

@misc{levicivita2022,
  title={Notion of Parallelism on a Generic Manifold and Consequent Geometrical Specification of the Riemannian Curvature}, 
  author={Tullio Levi-Civita},
  year={2022},
  eprint={2210.13239},
  archivePrefix={arXiv},
  primaryClass={gr-qc},
  url={https://arxiv.org/abs/2210.13239},
  note  = {Translated by Marco Godina and Julian Delens from the Italian edition of original Levi-Civita paper: {Rend. Circ. Mat. Palermo, (1917), Vol.42 (1), pp. 173-204}.}
}

@article{yang1954conservation,
  title = {Conservation of Isotopic Spin and Isotopic Gauge Invariance},
  author = {Yang, C. N. and Mills, R. L.},
  journal = {Phys. Rev.},
  volume = {96},
  issue = {1},
  pages = {191--195},
  numpages = {0},
  year = {1954},
  month = {Oct},
  publisher = {American Physical Society},
  doi = {10.1103/PhysRev.96.191},
  url = {https://link.aps.org/doi/10.1103/PhysRev.96.191}
}

@book{strogatz2024nonlinear,
  title={{Nonlinear dynamics and chaos: with applications to physics, biology, chemistry, and engineering}},
  author={Strogatz, Steven H},
  year={2024},
  publisher={Chapman and Hall, New York},
  location={New York}
}

@incollection{truesdell2004non,
  title={The non-linear field theories of mechanics},
  author={Truesdell, Clifford and Noll, Walter},
  booktitle={The non-linear field theories of mechanics},  
  year={2004},
  publisher={Springer Berlin, Heidelberg}
}

@article{onsager1931reciprocal,
  title = {{Reciprocal Relations in Irreversible Processes. I.}},
  author = {Onsager, Lars},
  journal = {Phys. Rev.},
  volume = {37},
  issue = {4},
  pages = {405--426},
  numpages = {0},
  year = {1931},
  month = {Feb},
  publisher = {American Physical Society},
  doi = {10.1103/PhysRev.37.405},
  url = {https://link.aps.org/doi/10.1103/PhysRev.37.405}
}

@article{onsager1931reciprocal2,
  title = {{Reciprocal Relations in Irreversible Processes. II.}},
  author = {Onsager, Lars},
  journal = {Phys. Rev.},
  volume = {38},
  issue = {12},
  pages = {2265--2279},
  numpages = {0},
  year = {1931},
  month = {Dec},
  publisher = {American Physical Society},
  doi = {10.1103/PhysRev.38.2265},
  url = {https://link.aps.org/doi/10.1103/PhysRev.38.2265}
}

@article{prigogine1967symmetry,
  title={On symmetry-breaking instabilities in dissipative systems},
  author={Prigogine, Ilya and Nicolis, Gr{\'e}goire},
  journal={J. Chem. Phys.},
  volume={46},
  number={9},
  pages={3542--3550},
  year={1967},
  publisher={American Institute of Physics},
  url={https://doi.org/10.1063/1.1841255}
}

@article{chorin2000optimal,
  title={Optimal prediction and the Mori--Zwanzig representation of irreversible processes},
  author={Chorin, Alexandre J and Hald, Ole H and Kupferman, Raz},
  journal={Proc. Natl. Acad. Sci. USA},
  volume={97},
  number={7},
  pages={2968--2973},
  year={2000},
  publisher={The National Academy of Sciences},
  url={https://doi.org/10.1073/pnas.97.7.2968}
}

@article{wheeler1945interaction,
  title = {Interaction with the Absorber as the Mechanism of Radiation},
  author = {Wheeler, John Archibald and Feynman, Richard Phillips},
  journal = {Rev. Mod. Phys.},
  volume = {17},
  issue = {2-3},
  pages = {157--181},
  numpages = {0},
  year = {1945},
  month = {Apr},
  publisher = {American Physical Society},
  doi = {10.1103/RevModPhys.17.157},
  url = {https://link.aps.org/doi/10.1103/RevModPhys.17.157}
}

@article{wheeler1949classical,
  title = {Classical Electrodynamics in Terms of Direct Interparticle Action},
  author = {Wheeler, John Archibald and Feynman, Richard Phillips},
  journal = {Rev. Mod. Phys.},
  volume = {21},
  issue = {3},
  pages = {425--433},
  numpages = {0},
  year = {1949},
  month = {Jul},
  publisher = {American Physical Society},
  doi = {10.1103/RevModPhys.21.425},
  url = {https://link.aps.org/doi/10.1103/RevModPhys.21.425}
}

@article{rovelli1991time,
  title = {Time in quantum gravity: An hypothesis},
  author = {Rovelli, Carlo},
  journal = {Phys. Rev. D},
  volume = {43},
  issue = {2},
  pages = {442--456},
  numpages = {0},
  year = {1991},
  month = {Jan},
  publisher = {American Physical Society},
  doi = {10.1103/PhysRevD.43.442},
  url = {https://link.aps.org/doi/10.1103/PhysRevD.43.442}
}

@article{wilson1974confinement,
  title = {Confinement of quarks},
  author = {Wilson, Kenneth G.},
  journal = {Phys. Rev. D},
  volume = {10},
  issue = {8},
  pages = {2445--2459},
  numpages = {0},
  year = {1974},
  month = {Oct},
  publisher = {American Physical Society},
  doi = {10.1103/PhysRevD.10.2445},
  url = {https://link.aps.org/doi/10.1103/PhysRevD.10.2445}
}

@article{kogut1979introduction,
  title = {An introduction to lattice gauge theory and spin systems},
  author = {Kogut, John B.},
  journal = {Rev. Mod. Phys.},
  volume = {51},
  issue = {4},
  pages = {659--713},
  numpages = {0},
  year = {1979},
  month = {Oct},
  publisher = {American Physical Society},
  doi = {10.1103/RevModPhys.51.659},
  url = {https://link.aps.org/doi/10.1103/RevModPhys.51.659}
}

@book{amari2000methods,
  title={Methods of information geometry},
  author={Amari, Shun-ichi and Nagaoka, Hiroshi},
  volume={191},
  year={2000},
  publisher={American Mathematical Society, Rhode Island, USA},
  location ={Rhode Island, USA}
}

@article{bronstein2017geometric,
  title={Geometric deep learning: going beyond euclidean data},
  author={Bronstein, Michael M and Bruna, Joan and LeCun, Yann and Szlam, Arthur and Vandergheynst, Pierre},
  journal={IEEE Signal Process. Mag.},
  volume={34},
  number={4},
  pages={18--42},
  year={2017},
  publisher={IEEE},
  url={https://doi.org/10.1109/MSP.2017.2693418}
}

@inproceedings{bruna2014spectral,
  author= {Joan Bruna and Wojciech Zaremba and Arthur Szlam and Yann LeCun},
  editor = {Yoshua Bengio and Yann LeCun},
  title = {Spectral Networks and Locally Connected Networks on Graphs},
  booktitle = {2nd International Conference on Learning Representations, {ICLR} 2014, Banff, AB, Canada, April 14-16, 2014, Conference Track Proceedings},
  year  = {2014},
  url  = {http://arxiv.org/abs/1312.6203}  
}

@article{auti2026hebbian,
  title = {Hebbian Physics Networks: A self-organizing computational architecture based on local physical laws},
  author = {Auti, Gunjan and Daiguji, Hirofumi and Tanaka, Gouhei},
  journal = {Phys. Rev. Res.},
  volume = {8},
  issue = {1},
  pages = {013309},
  numpages = {14},
  year = {2026},
  month = {Mar},
  publisher = {American Physical Society},
  doi = {10.1103/tzgk-jqj4},
  url = {https://link.aps.org/doi/10.1103/tzgk-jqj4}
}

@book{landau1987fluid,
  title={Fluid Mechanics},
  author={Landau, Lev Davidovich and Lifshitz, Evgeny Mikhailovich},
  volume={6},
  year={1987},
  publisher={Pergamon Press},
  address={Oxford},
  edition={2nd}
}

@article{marchetti2013hydrodynamics,
  title = {Hydrodynamics of soft active matter},
  author = {Marchetti, M. C. and Joanny, J. F. and Ramaswamy, S. and Liverpool, T. B. and Prost, J. and Rao, Madan and Simha, R. Aditi},
  journal = {Rev. Mod. Phys.},
  volume = {85},
  issue = {3},
  pages = {1143--1189},
  numpages = {0},
  year = {2013},
  month = {Jul},
  publisher = {American Physical Society},
  doi = {10.1103/RevModPhys.85.1143},
  url = {https://link.aps.org/doi/10.1103/RevModPhys.85.1143}
}

@book{de2013non,
  title={Non-equilibrium thermodynamics},
  author={De Groot, Sybren Ruurds and Mazur, Peter},
  year={2013},
  publisher={Dover Publications, New York}
}

@article{bazant2009towards,
  title={Towards an understanding of induced-charge electrokinetics at large applied voltages in concentrated solutions},
  author={Bazant, Martin Z and Kilic, Mustafa Sabri and Storey, Brian D and Ajdari, Armand},
  journal={Adv. Colloid Interface Sci.},
  volume={152},
  number={1-2},
  pages={48--88},
  year={2009},
  publisher={Elsevier},
  url={https://doi.org/10.1016/j.cis.2009.10.001}
}

@article{saintillan2013active,
  title={Active suspensions and their nonlinear models},
  author={Saintillan, David and Shelley, Michael J},
  journal={C. R. Phys.},
  volume={14},
  number={6},
  pages={497--517},
  year={2013},
  publisher={Elsevier},
  url={https://doi.org/10.1016/j.crhy.2013.04.001}
}

@article{rovelli2011forget,
  title={{“Forget time” Essay written for the FQXi contest on the Nature of Time}},
  author={Rovelli, Carlo},
  journal={Found. Phys.},
  volume={41},
  number={9},
  pages={1475--1490},
  year={2011},
  publisher={Springer},
  url={https://doi.org/10.1007/s10701-011-9561-4}
}

@article{cross1993pattern,
  title = {Pattern formation outside of equilibrium},
  author = {Cross, M. C. and Hohenberg, P. C.},
  journal = {Rev. Mod. Phys.},
  volume = {65},
  issue = {3},
  pages = {851--1112},
  numpages = {0},
  year = {1993},
  month = {Jul},
  publisher = {American Physical Society},
  doi = {10.1103/RevModPhys.65.851},
  url = {https://link.aps.org/doi/10.1103/RevModPhys.65.851}
}

@article{ramaswamy2010mechanics,
  title={The mechanics and statistics of active matter},
  author={Ramaswamy, Sriram},
  journal={Annu. Rev. Condens. Matter Phys.},
  volume={1},
  number={1},
  pages={323--345},
  year={2010},
  publisher={Annual Reviews},
  url={https://doi.org/10.1146/annurev-conmatphys-070909-104101}
}

@article{proca1936theorie,
  title={Sur la th{\'e}orie ondulatoire des {\'e}lectrons positifs et n{\'e}gatifs},
  author={Proca, AL},
  journal={J. Phys. Radium},
  volume={7},
  number={8},
  pages={347--353},
  year={1936},
  publisher={Soci{\'e}t{\'e} Fran{\c{c}}aise de Physique},
  url={https://doi.org/10.1051/jphysrad:0193600708034700}
}

@article{amari1991mathematical,
  title={Mathematical theory of neural learning},
  author={Amari, Shun-ichi},
  journal={New Generation Computing},
  volume={8},
  number={4},
  pages={281--294},
  year={1991},
  publisher={Springer},
  url={https://doi.org/10.1007/BF03037088}
}

@article{gerstner2018eligibility,
  title={Eligibility traces and plasticity on behavioral time scales: experimental support of neohebbian three-factor learning rules},
  author={Gerstner, Wulfram and Lehmann, Marco and Liakoni, Vasiliki and Corneil, Dane and Brea, Johanni},
  journal={Frontiers in neural circuits},
  volume={12},
  pages={53},
  year={2018},
  publisher={Frontiers Media SA},
  url={https://doi.org/10.3389/fncir.2018.00053}
}

@Inbook{hehl2017gauge,
    author={Hehl, Friedrich W.},
    editor={Lehmkuhl, Dennis and Schiemann, Gregor and Scholz, Erhard},
    title={Gauge Theory of Gravity and Spacetime},
    bookTitle={Towards a Theory of Spacetime Theories},
    year={2017},
    publisher={Springer New York},
    address={New York, NY},
    pages={145--169},
    isbn={978-1-4939-3210-8},
    doi={10.1007/978-1-4939-3210-8_5},
    url={https://doi.org/10.1007/978-1-4939-3210-8_5}
}

@misc{huang2023,
     author = {{Nvidia on-demand}},
     title = {{Fireside Chat with Ilya Sutskever and Jensen Huang: AI Today and Vision of the Future}},
     howpublished = {\url{https://www.nvidia.com/en-us/on-demand/session/gtcspring23-s52092}},
     year = {2023},   
}

@article{lecun2022path, 
  title={Introduction to latent variable energy-based models: a path toward autonomous machine intelligence},
  author={Dawid, Anna and LeCun, Yann},
  journal={Journal of Statistical Mechanics: Theory and Experiment},
  volume={2024},
  number={10},
  pages={104011},
  year={2024},
  publisher={IOP Publishing},
  url={https:/doi.org/10.1088/1742-5468/ad292b}
}

@misc{lecun2023,
  author = {Meta},
  title = {{I-JEPA: The first AI model based on Yann LeCun’s vision for more human-like AI}},
  howpublished = {\url{https://ai.meta.com/blog/yann-lecun-ai-model-i-jepa}},
  year = {2023},   
}

@article{mitchell2023ai,
  title={{AI’s challenge of understanding the world}},
  author={Mitchell, Melanie},
  journal={Science},
  volume={382},
  number={6671},
  pages={eadm8175},
  year={2023},
  publisher={American Association for the Advancement of Science},
  url={https://www.science.org/doi/10.1126/science.adm8175}
}

\end{document}